\documentclass[12pt]{article}

\usepackage[T1]{fontenc}
\usepackage[sc]{mathpazo}
\usepackage{amsmath}
\usepackage{amssymb}
\usepackage{enumerate}
\usepackage{amsthm}
\usepackage{amsfonts,mathrsfs}

\usepackage{hyperref}
\hypersetup{pdfpagemode=UseNone}

\newcommand{\tinyspace}{\mspace{1mu}}

\newcommand{\abs}[1]{\left\lvert\tinyspace #1 \tinyspace\right\rvert}

\newcommand{\setft}[1]{\mathrm{#1}}

\newcommand{\density}[1]{\setft{D}\left(#1\right)}

\def\I{\mathbb{1}}

\newenvironment{mylist}[1]{\begin{list}{}{
    \setlength{\leftmargin}{#1}
    \setlength{\rightmargin}{0mm}
    \setlength{\labelsep}{2mm}
    \setlength{\labelwidth}{8mm}
    \setlength{\itemsep}{0mm}}}
    {\end{list}}

\def\ot{\otimes}

\newcommand{\out}[2]{| #1\rangle\langle #2 |}

\newcommand{\Innerm}[3]{\left\langle #1 \left| #2 \right| #3 \right\rangle}

\newcommand{\Pa}[1]{\left(#1\right)}

\newcommand{\Br}[1]{\left[#1\right]}
\newcommand{\set}[1]{\{#1\}}
\newcommand{\Set}[1]{\left\{#1\right\}}

\newcommand{\ket}[1]{|#1\rangle}

\DeclareMathOperator{\trace}{Tr}

\newcommand{\Ptr}[2]{\trace_{#1}\Pa{#2}}

\newcommand{\Tr}[1]{\Ptr{}{#1}}

\def\cE{\mathcal{E}}
\def\cH{\mathcal{H}}

\def\rC{\mathrm{C}}
\def\rF{\mathrm{F}}

\def\rS{\mathrm{S}}

\newtheorem{thrm}{Theorem}[section]
\newtheorem{lem}[thrm]{Lemma}

\theoremstyle{definition}

\newtheorem{remark}[thrm]{Remark}

\numberwithin{equation}{section}

\newcounter{questionnumber}

\begin{document}

\title{\large\bf Universal upper bound for the Holevo information
induced by a quantum operation}

\author{Lin Zhang$^\text{a}$\footnote{E-mail: godyalin@163.com;
linyz@hdu.edu.cn},\; Junde Wu$^\text{b}$\footnote{E-mail: wjd@zju.edu.cn},\;
Shao-Ming Fei$^\text{c}$\footnote{ E-mail: feishm@mail.cnu.edu.cn}\\
  {\small $^\text{a}$\it Institute of Mathematics, Hangzhou Dianzi University, Hangzhou 310018, PR~China}\\
  {\small $^\text{b}$\it Department of Mathematics, Zhejiang University, Hangzhou 310027,
  PR~China}\\
  {\small $^\text{c}$\it School of Mathematical Sciences, Capital Normal
University, Beijing 100048, PR~China}}

\date{}
\maketitle \mbox{}\hrule\mbox\\
\begin{abstract}

Let $\cH_A\ot \cH_B$ be a bipartite system and $\rho_{AB}$ a quantum
state on $\cH_A\ot \cH_B$, $\rho_A = \Ptr{B}{\rho_{AB}}$, $\rho_B =
\Ptr{A}{\rho_{AB}}$. Then each quantum operation $\Phi_B$ on the quantum
system $\cH_B$ can induce a quantum ensemble
$\set{(p_\mu,\rho_{A,\mu})}$ on quantum system~$\cH_A$. In this
paper, we show that the Holevo quantity
$\chi\set{(p_\mu,\rho_{A,\mu})}$ of the quantum ensemble
$\set{(p_\mu,\rho_{A,\mu})}$ can be upper bounded by both subsystem entropies.
By using the result, we answer partly a conjecture of Fannes, de
Melo, Roga and
\.{Z}yczkowski.\\[0.1cm]


\end{abstract}
\mbox{}\hrule\mbox\\

\section{Introduction and preliminaries}

Let $\cH$ be a finite dimensional complex Hilbert space. A
\emph{quantum state} $\rho$ on $\cH$ is a positive semi-definite
operator of trace one, in particular, for each unit vector
$\ket{\psi} \in \cH$, the operator $\rho = \out{\psi}{\psi}$ is said
to be a \emph{pure state}. The set of all quantum states on $\cH$ is
denoted by $\density{\cH}$. For each quantum state
$\rho\in\density{\cH}$, its von Neumann entropy is defined by
$\rS(\rho) = - \Tr{\rho\log_2\rho}$. A \emph{quantum operation}
$\Phi$ on $\cH$ is a trace-preserving completely positive linear
mapping defined over the set $\density{\cH}$. It follows from
(\cite[Prop.~5.2 and Cor.~5.5]{Watrous}) that there exist linear
operators $\{M_\mu\}_{\mu=1}^K$ on $\cH$ such that $\sum_{\mu=1}^K
M^\dagger_\mu M_\mu = \I$ and $\Phi = \sum_\mu \mathrm{Ad}_{M_\mu}$,
that is, for each quantum state $\rho$, we have the Kraus
representation
\begin{eqnarray*}
\Phi(\rho) = \sum_{\mu=1}^K M_\mu \rho M^\dagger_\mu.
\end{eqnarray*}

\vskip0.1in

Let $\cE = \set{(p_\mu,\rho_\mu)}$ be a quantum ensemble on $\cH$,
that is, each $\rho_\mu\in \density{\cH}$, $p_\mu > 0$, and
$\sum_{\mu} p_\mu=1$.  The \emph{Holevo quantity} of the quantum
ensemble $\Set{\Pa{p_\mu,\rho_\mu}}$ is defined by the following
expression:
\begin{eqnarray}
\chi\Set{\Pa{p_\mu,\rho_\mu}} = \rS(\sum_\mu p_\mu \rho_\mu) -
\sum_\mu p_\mu \rS\Pa{\rho_\mu}.
\end{eqnarray}

\vskip0.1in

Let $\cH_A\ot \cH_B$ be a bipartite system and $\rho_{AB}$ a quantum
state on $\cH_A\ot \cH_B$, $\rho_A = \Ptr{B}{\rho_{AB}}$, $\rho_B =
\Ptr{A}{\rho_{AB}}$, $\Phi_B = \sum_\mu \mathrm{Ad}_{M_{B,\mu}}$ be
a quantum operation on quantum system $\cH_B$. Then $\Phi = \sum_\mu
\mathrm{Ad}_{\I_A\ot M_{B,\mu}}$ is a quantum operation on the
bipartite system $\cH_A\ot \cH_B$.

Let
$$
p_\mu = \Tr{\Pa{\I_A \ot M_{B,\mu}}\rho_{AB}\Pa{\I_A \ot
{M^\dagger_{B,\mu}} }}.
$$
Then $p_\mu\geqslant 0$ and $\sum_\mu p_\mu = 1$. Without loss of
generality, we assume that $p_\mu >0$. Let
\begin{eqnarray*}
\rho_{A,\mu} &\equiv& p_\mu^{-1}\Ptr{B}{\Pa{\I_A \ot
M_{B,\mu}}\rho_{AB}\Pa{\I_A \ot
{M^\dagger_{B,\mu}}}}\\
&=&p_\mu^{-1}\Ptr{B}{\Pa{\I_A \ot \sqrt{{M^\dagger_{B,\mu}}
M_{B,\mu}}}\rho_{AB} \Pa{\I_A \ot \sqrt{{M^\dagger_{B,\mu}}
M_{B,\mu}}}}.
\end{eqnarray*}
Then $\rho_{A,\mu}$ is a quantum state on $\cH_A$. Thus, quantum
operation $\Phi_B$ induced a quantum ensemble
$\set{(p_\mu,\rho_{A,\mu})}$ on quantum system $\cH_A$.

\vskip0.1in

In this Letter, the following result is obtained:

\begin{thrm}\label{th:main1}
$\chi\set{(p_\mu,\rho_{A,\mu})} \leqslant
\min\Set{\rS(\rho_A),\rS(\rho_B)}$.
\end{thrm}

By using this result, we answer partly a conjecture of Fannes, de
Melo, Roga and \.{Z}yczkowski.

\section{The proof of Theorem~\ref{th:main1}}
Clearly, $\chi\set{(p_\mu,\rho_{A,\mu})} \leqslant \rS(\rho_A)$ is
trivial by the definition of the Holevo information. It remains to
prove $\chi\set{(p_\mu,\rho_{A,\mu})} \leqslant \rS(\rho_B)$. The
nontrivial part of the proof is divided into three parts as follows:

\begin{enumerate}[(i)]
\item If $\set{\ket{\psi_{B,\mu}}}_{\mu=1}^K$ is a standard orthonormal
basis of $\cH_B$ and $M_\mu = \out{\psi_{B,\mu}}{\psi_{B,\mu}}$,
then it follows, from Theorem 3.1 in \cite{Zhang} and its proof,
that $\chi\set{(p_\mu,\rho_{A,\mu})} \leqslant \rS(\rho_B)$.

\item If $M_{B,\mu} = P_{B,\mu}$, where $P_{B,\mu}$ is a projector on $\cH_B$.
Note that $\sum_\mu P_{B,\mu} = \I_B$, so there is a standard
orthonormal basis $\set{\ket{u_{\mu,i}}}$ of $\cH_B$ such that
$$
P_{B,\mu} = \sum_i \out{u_{\mu,i}}{u_{\mu,i}}
$$
for each $\mu$.

Denote $p_{\mu,i} = \Innerm{u_{\mu,i}}{\rho_B}{u_{\mu,i}}$, without
loss of generality, we assume that $p_{\mu,i} >0 $, and denote
$$
\rho_{A,\mu, i} =
p_{\mu,i}^{-1}\Innerm{u_{\mu,i}}{\rho_{AB}}{u_{\mu,i}}.
$$
Thus
\begin{eqnarray}\label{eq:sumofprob}
p_\mu &=& \Tr{\Pa{\I_A \ot P_{B,\mu}}\rho_{AB}\Pa{\I_A \ot
P_{B,\mu}}} =
\Tr{\Pa{\I_A \ot P_{B,\mu}}\rho_{AB}}\nonumber\\
&=& \sum_i \Tr{\Innerm{u_{\mu,i}}{\rho_{AB}}{u_{\mu,i}}} =\sum_i
\Innerm{u_{\mu,i}}{\rho_B}{u_{\mu,i}} = \sum_i p_{\mu,i}
\end{eqnarray}
and
\begin{eqnarray}\label{eq:substate}
\Pa{\I_A \ot P_{B,\mu}}\rho_{AB}\Pa{\I_A \ot P_{B,\mu}} =
\sum_{i,i'} \Innerm{u_{\mu,i}}{\rho_{AB}}{u_{\mu,i'}} \ot
\out{u_{\mu,i}}{u_{\mu,i'}}.
\end{eqnarray}
It follows from Eq.~\eqref{eq:sumofprob} and Eq.~\eqref{eq:substate}
that
$$
p_\mu\rho_{A,\mu} = \sum_i p_{\mu,i}\rho_{A,\mu, i}.
$$
Therefore, by the concavity of von Neumann entropy, we have
$$
p_\mu\rS(\rho_{A,\mu}) \geqslant \sum_i p_{\mu,i}\rS(\rho_{A,\mu,
i}).
$$
So,
$$
\sum_\mu p_\mu\rS(\rho_{A,\mu}) \geqslant \sum_\mu\sum_i
p_{\mu,i}\rS(\rho_{A,\mu, i}).
$$
Thus, the desired inequality is obtained.

\item Now we prove the theorem generally. By the Naimark~theorem \cite{Watrous}, there exists a
quantum system $\cH_C$, a unit vector $\ket{0_C}\in\cH_C $ and a
projector $\set{P_{BC,\mu}}$ on the bipartite system $\cH_B\ot
\cH_C$ such that $\Innerm{0_C}{P_{BC,\mu}}{0_C}= M^\dagger_{B,\mu}
M_{B,\mu}$. Thus,
\begin{eqnarray*}
p_\mu \rho_{A,\mu} &=& \Ptr{B}{\Pa{\I_A \ot \sqrt{M^\dagger_{B,\mu}
M_{B,\mu}}}\rho_{AB}\Pa{\I_A \ot
\sqrt{M^\dagger_{B,\mu} M_{B,\mu}}}} \\
&=& \Ptr{BC}{\Pa{\I_A \ot P_{BC,\mu}}\Pa{\rho_{AB}\ot
\out{0}{0}_C}\Pa{\I_A \ot P_{BC,\mu}}}.
\end{eqnarray*}

So, the quantum ensemble $\set{(p_\mu,\rho_{A,\mu})}$ which is
induced by the quantum operation $\Phi_B$ can be considered as one
which is induced by the quantum operation $\Phi_{BC} = \sum_\mu
\mathrm{Ad}_{P_{BC,\mu}}$ over $\cH_B\ot \cH_C$. Thus, it follows
from (ii) that
$$
\chi\set{(p_\mu,\rho_{A,\mu})} \leqslant \rS(\rho_B\ot \out{0}{0}_C)
= \rS(\rho_B).
$$
\end{enumerate}

\section{The conjecture of Fannes, de Melo, Roga and \.{Z}yczkowski}

Let $\cE_N = \set{(p_i,\rho_i)}_{i=1}^N$ be a quantum ensemble on a
finite dimensional quantum system $\cH$, $\rF_{ij} =
\rF(\rho_i,\rho_j) = \Pa{\Tr{\abs{\sqrt{\rho_i}\sqrt{\rho_j}}}}^2$ be
the fidelity between $\rho_i$ and $\rho_j$. The matrix
$$
\rC_{\sqrt{\rF}}(\cE_N) = \Br{\sqrt{p_i p_j\rF_{ij}}}_{ij}
$$
is said to be a \emph{correlation matrix} of the quantum ensemble
$\cE_N = \set{(p_i,\rho_i)}_{i=1}^N$.

\vskip0.1in

For $N=2$ or $3$, the correlation matrix $\rC_{\sqrt{\rF}}(\cE_N) =
\Br{\sqrt{p_i p_j\rF_{ij}}}_{ij}$ is a legitimate state. However, if
$N\geqslant4$, then $\rC_{\sqrt{\rF}}(\cE_N) = \Br{\sqrt{p_i
p_j\rF_{ij}}}_{ij}$ fails to be a positive semi-definite matrix in
general \cite{Fannes}. For $N = 2$, the correlation matrix
$$
\rC_{\sqrt{\rF}}(\cE_2)=\Br{
        \begin{array}{cc}
          p_1 & \sqrt{p_1 p_2 \rF(\rho_1,\rho_2)} \\
          \sqrt{p_1 p_2 \rF(\rho_1,\rho_2)} & p_2 \\
        \end{array}
}
$$
was shown to satisfy the following inequality \cite{Roga}:
$$
\chi(\cE_2) \leqslant \rS(\rC_{\sqrt{\rF}}(\cE_2)).
$$
Moreover, the upper bound $\rS(\rC_{\sqrt{\rF}}(\cE_2))$ is the
tighter one in the above inequality.

\vskip0.1in

Fannes, de Melo, Roga and \.{Z}yczkowski conjectured that for $N=3$,
$ \chi(\cE_3) \leqslant \rS(\rC_{\sqrt{\rF}}(\cE_3))$ is also true
\cite{Fannes}.

\vskip0.1in

In what follows, we apply Theorem~\ref{th:main1} to answer partly
the conjecture.

\begin{lem}[\cite{Hou}]\label{3block} Let $\cH_1, \cH_2$ and $\cH_3$ be
three finite dimensional complex Hilbert spaces. Then the block
operator
$$
\left[
  \begin{array}{ccc}
    A & D & E\\
    D^\dagger & B & F\\
    E^\dagger &  F^\dagger & C\\
  \end{array}
\right]
$$
defined on $\cH_1\oplus\cH_2\oplus\cH_3$ is positive semi-definite
if and only if the following statements are valid:
\begin{enumerate}[(i)]
\item $A\geqslant0,B\geqslant0,C\geqslant0$;
\item there exist three contractive operators $R_1, R_2$ and $R_3$ such that
$D = \sqrt{A}R_1\sqrt{B}, F = \sqrt{B}R_2\sqrt{C}$, and {\small
$$
E = \sqrt{A}R_1\mathrm{supp}(B)R_2\sqrt{C} + \sqrt{A -
\sqrt{A}R_1\mathrm{supp}(B)R^\dagger_1\sqrt{A}}R_3\sqrt{C -
\sqrt{C}R^\dagger_2\mathrm{supp}(B)R^\dagger_2\sqrt{C}},
$$}
where $\mathrm{supp}(B)$ stands for the support projection of $B$.
\end{enumerate}
\end{lem}

\vskip0.1in

\begin{lem}\label{lem:block-pos}
Let $U, V$ and $W$ be three unitary operators on finite dimensional
complex Hilbert space $\cH$ and $\I$ be the identity operator on
$\cH$. Then the operator
$$
\Br{
  \begin{array}{ccc}
    \I & U & V \\
    U^\dagger & \I  & W \\
    V^\dagger & W^\dagger & \I  \\
  \end{array}
}
$$
is positive semi-definite if and only if $V = UW$.
\end{lem}


\begin{proof}
Taking $D = U, E = V, F = W$ and $A = B = C = \I$ in Lemma
\ref{3block}, we have that $R_1 = U, R_2 = W, \mathrm{supp}(B) = \I$
and $R_3$ is a contractive operator. Moreover, $V = UW$. That is
$$
\left[
  \begin{array}{ccc}
    \I & U & V \\
    U^\dagger & \I & W \\
    V^\dagger & W^\dagger & \I \\
  \end{array}
\right] \geqslant 0 \Longleftrightarrow V = UW.
$$
\end{proof}


\begin{remark}
The alternative proof of Lemma~\ref{lem:block-pos} may be given by
Theorem 3.1 in \cite{Orlov}.
\end{remark}

\begin{thrm}\label{th:main2}

Let $\cE_3 = \set{(p_1,\rho_1),(p_2,\rho_2),(p_3,\rho_3)}$ be a
quantum ensemble on the finite dimensional quantum system $\cH$. It
follows from the polar decomposition theorem that there exist three
unitary operators $V, U $ and $W$ on $\cH$ such that
\begin{eqnarray*}
\abs{\sqrt{\rho_2}\sqrt{\rho_1}} &=& U\sqrt{\rho_2}\sqrt{\rho_1}, \\
\abs{\sqrt{\rho_3}\sqrt{\rho_1}} &=& V\sqrt{\rho_3}\sqrt{\rho_1}, \\
\abs{\sqrt{\rho_3}\sqrt{\rho_2}} &=& W\sqrt{\rho_3}\sqrt{\rho_2}.
\end{eqnarray*} If $V = UW$, then
$$
\chi(\cE_3)\leqslant \rS(\rC_{\sqrt{\rF}}(\cE_3)).
$$
\end{thrm}

\begin{proof}
By the conditions, it follows that
\begin{eqnarray*}
\Tr{\sqrt{\rho_1}U\sqrt{\rho_2}} &=& \sqrt{\rF_{12}},\\
\Tr{\sqrt{\rho_1}V\sqrt{\rho_3}} &=& \sqrt{\rF_{13}},\\
\Tr{\sqrt{\rho_2}W\sqrt{\rho_3}} &=& \sqrt{\rF_{23}}.
\end{eqnarray*}

Let $\cH_A=\cH$, $\cH_B=\mathbb{C}^3$, and
$$
\rho_{AB} = \Br{
               \begin{array}{ccc}
                 p_1\rho_1 & \sqrt{p_1p_2}\sqrt{\rho_1}U\sqrt{\rho_2} & \sqrt{p_1p_3}\sqrt{\rho_1}V\sqrt{\rho_3} \\
                 \sqrt{p_1p_2}\sqrt{\rho_2}U^\dagger\sqrt{\rho_1} & p_2\rho_2 & \sqrt{p_2p_3}\sqrt{\rho_2}W\sqrt{\rho_3} \\
                 \sqrt{p_1p_3}\sqrt{\rho_3}V^\dagger\sqrt{\rho_1} & \sqrt{p_2p_3}\sqrt{\rho_3}W^\dagger\sqrt{\rho_2} & p_3\rho_3 \\
               \end{array}
             }.
$$
Now, we only need to show that $\rho_{AB}$ is a positive
semi-definite operator on $\cH_A\ot\cH_B$. Note that
$$ \rho_{AB} =\Br{
               \begin{array}{ccc}
                 \sqrt{p_1\rho_1} & 0 & 0 \\
                 0 & \sqrt{p_2\rho_2} & 0 \\
                 0 & 0 & \sqrt{p_3\rho_3} \\
               \end{array}
             }\Br{
               \begin{array}{ccc}
                 \I & U & V \\
                 U^\dagger & \I & W \\
                 V^\dagger & W^\dagger & \I \\
               \end{array}
             }\Br{
               \begin{array}{ccc}
                 \sqrt{p_1\rho_1} & 0 & 0 \\
                 0 & \sqrt{p_2\rho_2} & 0 \\
                 0 & 0 & \sqrt{p_3\rho_3} \\
               \end{array}
             },
$$
and it follows from Lemma~\ref{lem:block-pos} that $\rho_{AB}
\geqslant 0$ is equivalent to
$$ \Br{
               \begin{array}{ccc}
                 \I & U & V \\
                 U^\dagger & \I & W \\
                 V^\dagger & W^\dagger & \I \\
               \end{array}
             } \geqslant 0 \Longleftrightarrow V = UW.
$$
Moreover, it is easy to show that
$$
\rho_A = \Ptr{B}{\rho_{AB}} = \sum_{i=1}^3 p_i\rho_i,\quad \rho_B =
\rC_{\sqrt{\rF}}(\cE_3).
$$
Since $\dim(\cH_B) = 3$, take a standard orthogonal basis
$\set{\ket{\mu_B}}$ of $\cH_B$ such that $p_\mu \rho_{A,\mu} =
\Innerm{\mu_B}{\rho_{AB}}{\mu_B}$. By Theorem~\ref{th:main1}, we
have
$$
\chi(\cE_3) = \chi\set{(p_\mu,\rho_{A,\mu})}\leqslant \rS\Pa{\rho_B}
= \rS(\rC_{\sqrt{\rF}}(\cE_3)).
$$
This completes the proof.
\end{proof}

\begin{remark}
In fact, Lemma~\ref{lem:block-pos} can be easily generalized to the
case where 3-by-3 block matrix is replaced by $K$-by-$K$
($K\geqslant 3$) block matrix of unitary entries. The generalization
is described as follows:

\vskip0.1in

Assume that the following $K\times K$ block matrix of unitary
entries is positive semi-definite:
$$
\left[
              \begin{array}{cccc}
                U_{11} & U_{12} & \cdots & U_{1K} \\
                U_{21} & U_{22} & \cdots & U_{2K} \\
                \vdots & \vdots & \ddots & \vdots \\
                U_{K1} & U_{K2} & \cdots & U_{KK} \\
              \end{array}
            \right]\equiv P\geqslant0.
$$
Then these unitary operators satisfy the conditions:
\begin{enumerate}
\item[$\bullet$] $U_{ii} = \I$ for each index $i$; $U_{ji} = U^\dagger_{ij}$ for all indices $i,j$.
\end{enumerate}
Thus
$$
P = \left[
              \begin{array}{cccc}
                \I & U_{12} & \cdots & U_{1K} \\
                U^\dagger_{12} & \I & \cdots & U_{2K} \\
                \vdots & \vdots & \ddots & \vdots \\
                U^\dagger_{1K} & U^\dagger_{2K} & \ldots & \I \\
              \end{array}
            \right].
$$
Furthermore, we have that $P$ is of the following forms:
\begin{enumerate}[(a)]
\item{\small
$$
P = \left[
  \begin{array}{ccccccc}
    \I & U_1 & U_1U_2 & U_1U_2U_3 & \ldots & \ldots & U_1U_2\cdots U_{K-1} \\
    U^\dagger_1 & \I & U_2 & U_2U_3 & U_2U_3U_4 & \ddots & \vdots \\
    U^\dagger_2U^\dagger_1 & U^\dagger_2 & \I & U_3 & U_3U_4 & \ddots & \vdots \\
    U^\dagger_3U^\dagger_2U^\dagger_1 & U^\dagger_3U^\dagger_2 & U^\dagger_3 & \I & \ddots & \ddots &  U_{K-3}U_{K-2}U_{K-1} \\
    \vdots & \ddots & \ddots & \ldots & \ddots & \ddots & U_{K-2}U_{K-1} \\
    \vdots & \ddots & \ddots & \ldots & \ddots & \ddots & U_{K-1} \\
     U^\dagger_{K-1}\cdots U^\dagger_2U^\dagger_1& \ldots & \ldots & U^\dagger_{K-1}U^\dagger_{K-2}U^\dagger_{K-3} & U^\dagger_{K-1}U^\dagger_{K-2} & U^\dagger_{K-1} & \I \\
  \end{array}
\right]
$$
} for a collection of unitary operators $\Set{U_i: i =
1,\ldots,K-1}$ on $\cH$,
\end{enumerate}
or
\begin{enumerate}[(b)]
\item
$$
P = \left[
      \begin{array}{c}
        V_1 \\
        V_2 \\
        \vdots \\
        V_K \\
      \end{array}
    \right]\left[
             \begin{array}{cccc}
               V^\dagger_1 & V^\dagger_2 & \cdots & V^\dagger_K \\
             \end{array}
           \right]
$$
for a collection of unitary operators $\Set{V_i: i = 1,\ldots,K}$ on
$\cH$.
\end{enumerate}

The outline of the proof is the following. The fact that $P$ is of
the form (a) can be easily derived by applying repeatedly the
Theorem 3.1 in \cite{Orlov} to a block matrix. Indeed, we first
apply it to the new block matrix:
$$
\Br{\begin{array}{cc}
      \I & X \\
      X^\dagger & A
    \end{array}
}\geqslant0,
$$
where $X = \Br{U_{12},\ldots,U_{1K}}$ and
$$
A = \Br{\begin{array}{cccc}
                \I & U_{23} & \cdots & U_{2K} \\
                U^\dagger_{23} & \I & \cdots & U_{3K} \\
                \vdots & \vdots & \ddots & \vdots \\
                U^\dagger_{2K} & U^\dagger_{3K} & \ldots & \I \\
              \end{array}}.
$$
Then apply it again to a similar block structure for $A$, and so on.
Finally we obtain the form (a) of $P$. The forms (a) and (b) are
equivalent via the following identification:
$$
U_1 = V_1V_2^\dagger,U_2=V_2V_3^\dagger,
\ldots,U_{K-1}=V_{K-1}V_K^\dagger.
$$
\end{remark}

\section{Concluding remarks}

In this Letter, we obtained a universal upper bound for the Holevo
quantity which is induced by a quantum operation and proved that for
a given quantum ensemble which consists of $N$ quantum states on the
same space, a so-called correlation matrix $\rC_{\sqrt{\rF}}(\cE_N)$
can be constructed. Its von Neumann entropy is shown to be an upper
bound of the Holevo quantity for $N=3$ under some constraints. We
also generalized Lemma~\ref{lem:block-pos} and obtained an
interesting characterization of positivity of special operator
matrix, which may shed new light on solving other related problems
in quantum information theory.

\subsection*{Acknowledgement} 

The authors wish to express their thanks to the referees for their valuable comments and suggestions. This project is supported by the Natural Science Foundation of China (Grant Nos. 11171301, 10771191, 10471124 and 11275131) and Natural Science Foundation of Zhejiang Province of China (Grant Nos. Y6090105).




\end{document}